\documentclass{cta-author}

\usepackage{amsthm}
\usepackage{amsmath,amssymb,amsfonts}

\usepackage{xcolor}
\usepackage{graphicx}

\usepackage{subcaption}

\newcommand{\ml}{}
\newcommand{\dan}{}
\newcommand{\mml}{}

\usepackage{varioref}
\usepackage{hyperref}
\usepackage[capitalise]{cleveref}

\newtheorem{theorem}{Theorem}
\newtheorem{lemma}{Lemma}

\newtheorem{remark}{Remark}

\begin{document}

\supertitle{Brief Paper}

\title{Stochastic Optimal Investment Strategy for Net-Zero Energy Houses}

\author{\au{Mengmou Li $^1$} \au{Taichi Tanaka $^1$} \au{A. Daniel Carnerero $^1$} \au{Yasuaki Wasa $^2$} \au{Kenji Hirata $^3$} \au{Yasumasa Fujisaki $^4$} \au{Yoshiaki Ushifusa $^5$} \au{Takeshi Hatanaka $^1$}  }

\address{\add{1}{Department of Systems and Control Engineering, School of Engineering, Tokyo Institute of Technology, Tokyo 152-8550, Japan.}
\add{2}{Department of Electrical Engineering and Bioscience, School of Advanced Science and Engineering, Faculty of Science and Engineering, Waseda University, Tokyo 169-8555, Japan.}
\add{3}{Department of Electrical and Electronic Engineering, School of Engineering, University of Toyama,
Toyama 930-8555, Japan.}
\add{4}{Department of Information and Physical Sciences, Graduate School of Information Science and Technology, Osaka University, Osaka 565 0871, Japan}
\add{5}{Faculty of Economics and Business Administration, The University of Kitakyushu, Fukuoka 802-8577, Japan.}
\email{ \{li.m; tanaka; carnerero\}@hfg.sc.e.titech.ac.jp; wasa@waseda.jp; hirata@eng.u-toyama.ac.jp; fujisaki@ist.osaka-u.ac.jp; ushifusa@kitakyu-u.ac.jp; hatanaka@sc.e.titech.ac.jp.}}

\begin{abstract}
\looseness=-1 
In this research, we investigate Net-Zero Energy Houses (ZEH), which harness regionally produced electricity from photovoltaic (PV) panels and fuel cells,  integrating them into a local power system in pursuit of achieving carbon neutrality. This paper examines the impact of electricity sharing among users who are working towards attaining ZEH status through the integration of PV panels and battery storage devices. We propose two potential scenarios: the first assumes that all users individually invest in storage devices, hence minimizing their costs on a local level without energy sharing; the second envisions cost minimization through the collective use of a shared storage device, managed by a central manager. These two scenarios are formulated as a stochastic convex optimization and a cooperative game, respectively.  To tackle the stochastic challenges posed by multiple random variables, we apply the Monte Carlo sample average approximation (SAA) to the problems. To demonstrate the practical applicability of these models, we implement the proposed scenarios in the Jono neighborhood in Kitakyushu, Japan.
\end{abstract}

\maketitle

\section{Introduction}
Addressing the urgent and formidable challenge of global warming by reducing $\text{CO}_2$ emissions has become a paramount concern, largely due to the detrimental effects of climate change on both the environment and human society. In 2019, one-third of Japan's electricity was generated from coal, and fossil fuels accounted for $88\%$ of Japan's power supply \cite{Country2020,Japan2021}. As the Japanese government aims for carbon neutrality by 2050 \cite{CarbonNeu}, renewable energy sources like solar and wind power are gaining importance as alternative fuel sources. Although renewable energy constituted less than $10\%$ of Japan's total energy consumption in 2019, the decreasing costs of solar and wind power, along with the ongoing economic recovery from COVID-19, are expected to increase the share of renewable energy in Japan's energy mix in the coming years \cite{Country2020}.

As the world strives for carbon neutrality and the integration of renewable energy sources, the concept of Net-Zero Energy Houses (ZEH) is getting increasing attention as a technology to reduce $\text{CO}_2$ emissions at home. A ZEH is a house that has an annual net energy consumption of around zero
by maximizing energy savings while maintaining a comfortable living environment. 
Achieving an annual net energy consumption of around zero is made possible through improved heat insulation, high-efficiency equipment, and photovoltaic (PV) power generation \cite{efficiency2015definition}.
To attain ZEH status, devices such as PV panels, fuel cells, and battery storage devices must be integrated into a regional power system \cite{wu2021residential}. However, it is vital to quantitatively assess the effects of introducing these devices to ensure the safe penetration of ZEH into the power grid.

In general, PV panels do not generate stable power output. The incorporation of battery storage can help mitigate these power fluctuations by effectively charging and discharging as needed. The integration of PV systems and battery storage into power grids has been a research focus for many years, and numerous studies have been conducted, including notable reviews such as \cite{hoppmann2014economic,azuatalam2019energy,khezri2022optimal}. 
In \cite{hoppmann2014economic}, the economic aspect of battery storage is reviewed and a techno-economic model for calculating the profitability of investing in battery storage is proposed for residential PV systems in Germany. 
In \cite{azuatalam2019energy}, a systematic review of energy management strategies for PV-battery systems under various scenarios is conducted. 
In \cite{khezri2022optimal}, a comprehensive review of key parameters to be considered for optimal planning of PV panels and battery storage systems is provided.
In particular, various optimization-based analyses have been conducted on the design and management of PV panels and battery storage systems under different scenarios, including those by \cite{pham2009optimal,zhu2014optimal,ratnam2015optimization, khalilpour2016planning,bordin2017linear,yan2018optimized}.
Ref.~\cite{pham2009optimal} considers the optimal operation of a PV-battery energy system using anticipation and reactive management.
Ref.~\cite{zhu2014optimal} proposes an optimal design and management strategy for an energy system with PV and battery storage to reduce energy costs.
Ref.~\cite{khalilpour2016planning} introduced a multi-period mixed-integer linear program to solve the investment problem of PV-battery systems.
Ref.~\cite{bordin2017linear} considers battery degradation costs in the energy management of off-grid systems. 
Ref.~\cite{yan2018optimized} focused on the energy management of a commercial EV charging station with integrated PV generation and battery storage. 

In recent years, sharing economy models via battery storage have become crucial for managing energy and reducing electricity costs in regional power systems \cite{wang2013active,kalathil2017sharing,parra2017interdisciplinary,chakraborty2018sharing,wang2019incentive,henni2021sharing}.
An energy management strategy for demand response using shared battery storage is proposed in \cite{wang2013active}.
Ref.~\cite{kalathil2017sharing} formulates a cooperative game for the optimal battery storage investment of a group of users sharing their electricity storage, considering fluctuations in power prices, while \cite{chakraborty2018sharing} uses cooperative game theory to evaluate scenarios where users invest in battery storage individually or jointly. 
Ref.~\cite{wang2019incentive} presents an incentive design to encourage users to share power obtained from PV generation in local power systems. 
In \cite{henni2021sharing}, a sharing economy model allowing residential communities to share excess PV generation, stored electricity, and storage capacity is proposed. 
However, these studies do not specifically address the optimal investment strategy to attain ZEH status. Numerous critical studies on Net-Zero Energy Buildings (ZEB) have been presented in \cite{vieira2017energy,harkouss2019optimal,ahmed2022assessment} and references therein. While these studies provide invaluable insights, they do not consider the optimal strategies for investment in energy sharing among users, which requires further investigation.

\subsection{A motivating case study}\label{subsection jono}
Jono, a neighborhood in Kitakyushu, Japan, serves as an ideal case study for our research. This community consists of $140$ households, which together form a distinct regional power system connected to the main power grid. These households have installed rooftop PV panels.
Fig.~\ref{fig:Power consumption X} and \ref{fig:Power generation Y}, display the daily power consumption and generation for one specific household.
The figures highlight the time-varying nature of power generation from PV panels, making it challenging to maintain a power balance within the system. However, using batteries to store excess energy during peak solar generation times and drawing upon this stored energy when solar panels produce less can potentially mitigate this power imbalance and attain ZEH status within the regional power system.
\begin{figure}[t]
     \centering
     \begin{subfigure}[b]{1\linewidth}
         \centering
         \includegraphics[width = 1\linewidth]{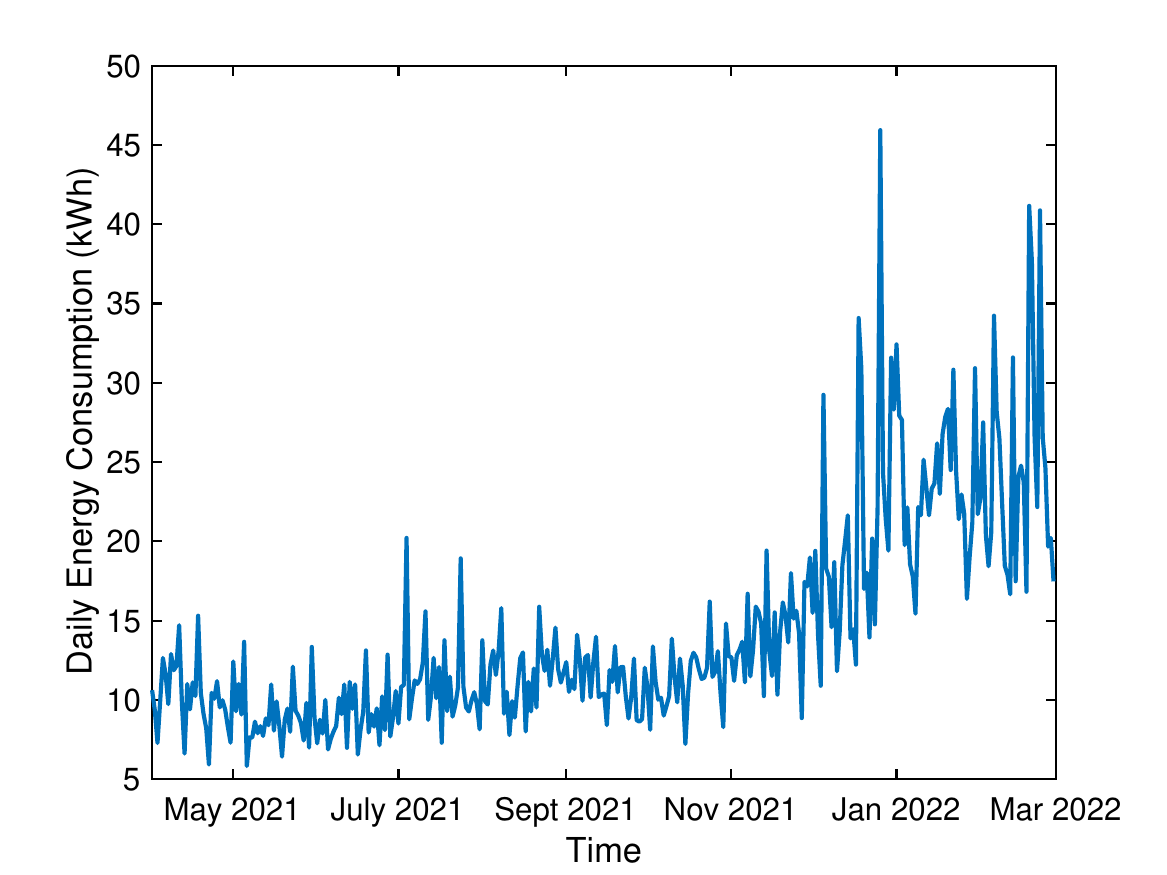}
         \caption{Daily energy consumption.}
         \label{fig:Power consumption X}
     \end{subfigure}
     \hfill
     \begin{subfigure}[b]{1\linewidth}
         \centering
         \includegraphics[width = 1\linewidth]{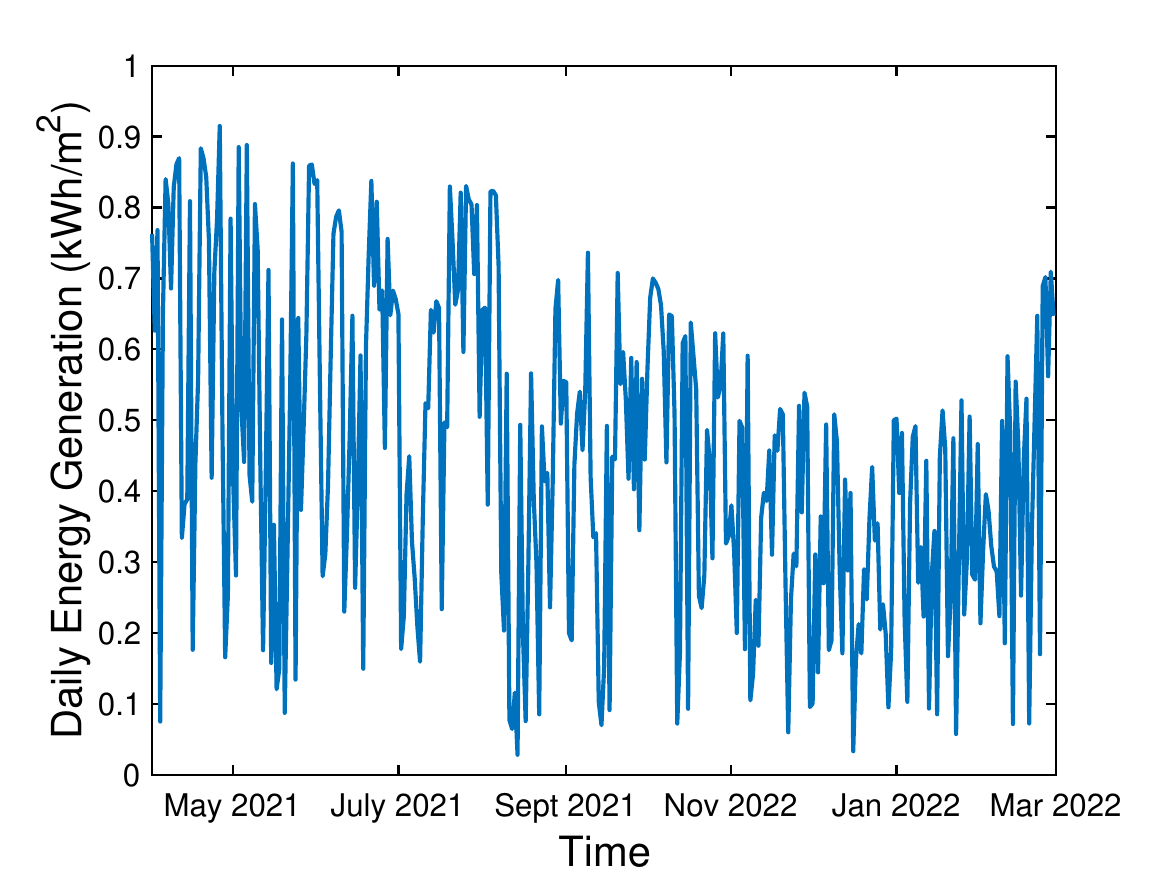}
         \caption{Daily power generation from PV panels per $\text{m}^2$.}
         \label{fig:Power generation Y}
     \end{subfigure}
        \caption{The daily power consumption and power generation from PV panels per $\text{m}^2$ for one household over 334 days.}
        \label{fig:Power consumption and generation}
\end{figure}
This example also suggests that power generation from photovoltaic systems can be modeled using random variables based on existing data.
 We will investigate the optimal investment problem of PV panels and battery storage devices from an economic standpoint.

\subsection{Contributions}
In this paper, we propose an optimal investment strategy for PV panels and battery storage devices within a regional power system aiming to attain ZEH status for users, with a focus on electricity sharing.
We consider two scenarios involving users and a manager in the regional power system. We examine two scenarios within the regional power system, involving users and a manager. The first scenario focuses on individual investment in battery storage, while the second scenario revolves around joint investment in a shared storage solution. These scenarios are formulated as stochastic convex optimization and cooperative game problems respectively, with cost functions demonstrated to be convex.
{As the proposed investment problems are of stochastic nature and thus many random variables are present (PV power generation, daily consumption, etc.), we employ the Monte Carlo sample average approximation approach, which allows us to generate solvable deterministic convex problems.}
\ml{It is important to note that in the application to the Jono neighborhood, the sampling data generated is derived from real historical data on PV generation and power consumption.}

\mml{
Compared to the aforementioned literature on PV-battery power systems, our research specifically centers around the optimization-based investment of PV and battery storage, coupled with power sharing for attaining Net-Zero Energy Houses. Notably, our approach is distinct as it utilizes real data collected from an entire neighborhood, comprising multiple users.
In contrast to the formulations of power sharing presented in \cite{kalathil2017sharing} and \cite{chakraborty2018sharing}, our approach incorporates the investment considerations for PV panels. Moreover, by introducing a manager into the framework, we formulate a game with convex functions that are globally solvable, preferable to nonconvex functions employed in \cite{kalathil2017sharing}.}

The remainder of the paper is organized as follows. In \Cref{Section Optimal Investment Problems}, we formulate the two scenarios for optimal investment in PV panels and battery storage. In \Cref{Section Stochastic optimal investment strategy}, we apply the Monte Carlo sample average approximation (SAA) approach to solve the stochastic problems. In \Cref{Section A case study}, we discuss a case study in Jono using the two scenarios. Finally, we conclude the paper in \Cref{Section Conclusion}.

\section{Optimal Investment Problems}\label{Section Optimal Investment Problems}

\ml{Notation: $\mathbb{E} [\cdot]$ and $\operatorname{\mathbb{V}ar} [\cdot]$ represent the expected value and variance of random variables. $\|\cdot\|$ denotes the max-norm of a vector, i.e., $\|x\| = \max \{x_1, \ldots, x_n\}$, where $x = (x_1, \ldots, x_n) \in \mathbb{R}^{n}$.}

Let us consider a regional power grid consisting of $n$ households and one manager. We assume that users in this regional power system generate electricity using rooftop PV panels and fuel cells, and may also exchange electricity with the manager. The manager is connected to the main power grid, and exchanges power with both individual users and the power grid.
The goal of ZEH is to balance each user's energy consumption with the energy generated by PV power generation, fuel cells, and other sources, ultimately reducing the net energy consumption to zero. Assuming that reducing the net energy consumption of the entire region to zero is practically achievable using ZEH, we propose the following two formulations regarding investment plans for PV panels and battery storage devices:
\begin{enumerate}
    \item In the first scenario, we assume that all users invest in individual storage devices and minimize their costs locally without sharing energy.
    \item In the second scenario, all users aim to minimize their own costs while sharing a common joint storage device, the capacity of which is distributed to users by the manager.
\end{enumerate}
The relationships among users, the manager, and the main power grid for the two scenarios are illustrated in \Cref{fig:Individual_Battery,fig:User_manager}, respectively.
\begin{figure}[htbp]
    \centering
    \begin{subfigure}[b]{1\linewidth}
         \centering
         \includegraphics[width = 1\linewidth]{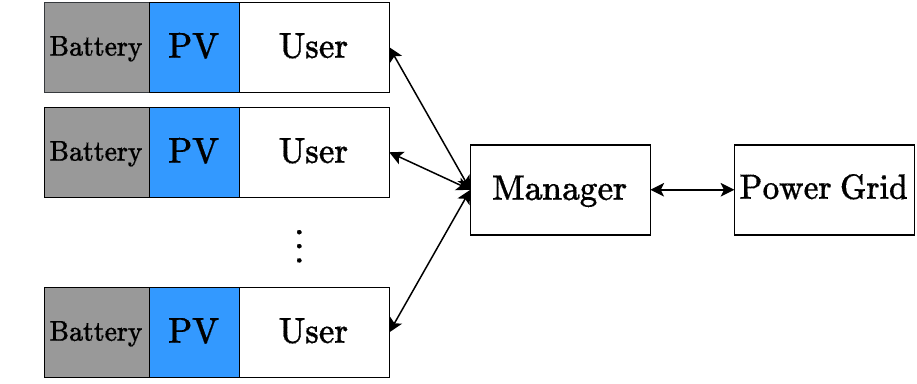}
         \caption{\dan{First scenario:} users have individual battery storage devices.}
         \label{fig:Individual_Battery}
    \end{subfigure}
    \begin{subfigure}[b]{1\linewidth}
         \centering
         \includegraphics[width = 1\linewidth]{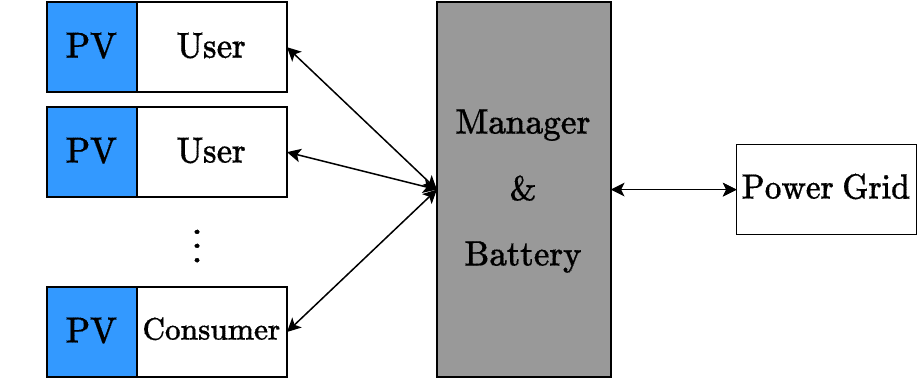}
         \caption{\dan{Second scenario: }users share a common  storage device.}
         \label{fig:User_manager}
    \end{subfigure}
    \caption{\dan{Users, manager and the power grid in both scenarios.}}
    \label{fig:my_label}
\end{figure}
\subsection{Individual storage investment}
In this subsection, we examine the scenario in which individual users independently achieve a balance between electricity supply and demand. We propose a stochastic optimization problem for the investment strategy of PV panels and battery storage devices.

In this scenario, we assume that user $i$ will invest in PV panels, battery storage, and a household fuel cell to determine an optimal energy management strategy for $T$ days into the future. Let $X_i (t) ~ \mathrm{[kWh]}$ be the power consumption of user $i$ on the $t$-th day, and $Y_i (t) ~\mathrm{[kWh/m^2]}$ be the power generation per unit area of the PV panels. Both $X_i (t)$ and $Y_i(t)$ are random variables with known distribution functions. 
Then, user $i$ decides to invest a total area of $a_i ~\mathrm{[m^2]}$  for the photovoltaic panels and a total of $C_i~\mathrm{[kWh]}$ for battery storage. 
 
First, let us consider the case where the power consumption $X_i (t)$ cannot be compensated by the power generation $a_i Y_i (t)$ from the PV power generation. In this situation, the user will compensate for the shortage using the remaining capacity of the battery storage, and if the remaining storage is insufficient, they will resort to a household fuel cell. In this case, the purchase cost of gas used for power generation becomes a penalty for not being able to supplement power. Therefore, the cost of power generation by household fuel cells when PV power generation is insufficient is given by
\begin{align}\label{eq: cost of gas individual}
J_{\text{gas}, i} (a_i, C_i, t) = \pi_{\text{gas}} \mathbb{E} \left[  \max  \left\{ X_i (t)  -  a_i Y_i (t) - \beta_i (t)  C_i, 0  \right\} \right],
\end{align}
where $\pi_{\text{gas}} ~ \mathrm{[ \text{\textyen} /kWh]}$ represents the power generation cost per unit of electricity from household fuel cells with \text{\textyen} representing JPY, and the power in the battery storage at the beginning of day $t$ is given by $\beta_i (t) C_i$, with $\beta_i (t) \in [0,1]$.

Next, let us consider the case where the power from the PV power generation  $a_i Y_i (t)$ exceeds the power consumption $X_i (t)$. At this time,  if the residual power exceeds the allowable amount of charge of the battery storage, the surplus power must flow to the main power grid.  However, reverse power flow is an action that puts a burden on the power grid.  Therefore, we try to suppress it by imposing a penalty according to the reverse flow rate.
The cost for the excess of PV power generation is given by
\begin{align}\label{eq: cost of rev individual}
& J_{\text{rev}, i}  (a_i, C_i, t) \nonumber\\
= & \pi_{\text{rev}} \mathbb{E} \left[  \max  \left\{   a_i Y_i (t) - X_i (t) -   \left( 1  -  \beta_i (t) \right) C_i, 0 \right\} \right],
\end{align}
where $\pi_{\text{rev}} ~ \mathrm{[ \text{\textyen}/kWh]}$ represents the price for reverse power flow per unit power.

Lastly, the cost of investing PV panels and battery storage should be considered. Let $\pi_{\text{PV}} ~ \mathrm{[ \text{\textyen}  /m^2 ]}$ denote the capital cost per unit area of the PV panels amortized over its lifespan, and $\pi_{\text{B}} ~ \mathrm{[ \text{\textyen} /kWh]}$ the capital cost per unit capacity of the battery storage, amortized over its lifespan. Then, the cost for the one-time capital investment is given by
\begin{align}\label{eq: capital cost individual}
J_{\text{int}, i} (a_i, C_i)= a_i \pi_{\text{PV}} + C_i \pi_{\text{B}}.
\end{align}
Summing up the above terms \eqref{eq: cost of gas individual}--\eqref{eq: capital cost individual},  we obtain the cost function for user $i$ to achieve the power balance between demand and supply. The investment strategy is formulated as an optimization problem,
\begin{align}\label{stochastic problem individual users}
& \min_{a_i, C_i} \left\{ \vphantom{\sum_{i= 1}^T} J_i (a_i, C_i) = J_{\text{int}, i} (a_i, C_i) \right. \nonumber \\
& \quad \left.  + \sum_{t = 1}^{T} \left\{ J_{\text{gas}, i} (a_i, C_i, t) + J_{\text{rev}, i} (a_i, C_i, t) \right\} \right\}
\end{align}
\ml{where the arguments of $a_i$, $C_i$ on the right-hand side are omitted for ease of notation.}

\subsection{Joint storage investment with storage distribution}
The investment strategy \ml{in the last subsection} can be readily extended to the global investment in PV panels and battery storage devices following similar arguments.
Specifically, the optimization problem for global investment is given by 
\begin{align}\label{eq:global_optimization_costfun}
    \min_{a_1,\ldots,a_n, C_{\text{a}}}  \hspace{-1mm} \left\{ J (a_1, \ldots, a_n, C_{\text{a}}) \hspace{-1mm} = J_{\text{int}} + \sum_{t = 1}^{T} \{ J_{\text{gas}} (t) + J_{\text{rev}} (t)\} \right\}
\end{align}
where 
\begin{align}
    J_{\text{int}} (a, C_{\text{a}})= \sum_{i}^{n} a_i \pi_{\text{PV}} + C_{\text{a}} \pi_{\text{B}},
\end{align}
\begin{align}
    J_{\text{gas}} (t) = \pi_{\text{gas}} \mathbb{E} \left[ \max \left\{ \sum_{i = 1}^{n} \left\{ X_i (t) - a_i Y_i (t) - \beta_i (t) C_{\text{a}} \right\}, 0 \right\} \right],
\end{align}
\begin{align}\label{eq:global_optimization_Jrev}
    & J_{\text{rev}} (t) \nonumber \\
    = & \pi_{\text{rev}} \mathbb{E} \left[ \max \left\{ \sum_{i = 1}^{n} \left\{ a_i Y_i (t) - X_i (t) - (1 - \beta_i (t)) C_{\text{a}}\right\}, 0 \right\} \right],
\end{align}
$a = (a_1, \ldots, a_n)$, and $C_a$ denotes the joint storage to invest. 
However, the storage distribution for users is unclear from the solution to the above optimization. Therefore, a scheme containing the dispatch of the storage devices should be introduced.

In this subsection, we explore another scenario where all users in the region share joint battery storage through a manager and achieve a power balance between demand and supply. In this scenario, user $i$ invests in PV panels and a household fuel cell, while the manager invests in joint battery storage and trades power with the user. The definitions of $T$, $X_i(t)$, and $Y_i (t)$ remain the same as in the previous subsection. 
User $i$ decides on the area $a_i ~ \mathrm{[m^2]}$ of PV panels to invest while the manager determines the capacity $C_{\text{a}} ~ \mathrm{[ kWh]}$ for the battery storage investment. Furthermore, the manager is responsible for making adjustments to achieve the power balance. The overall cost consists of two parts: one from the users and the other from the manager.

First, let us consider the cost borne by the users. The users aim to minimize their cost by designing the area of PV panels $a_i$ based on the capacity $C_{\text{a}}$ provided by the manager. In this scenario, when the power generated by user $i$'s PV panels exceeds their power consumption $X_i (t)$, user $i$ can earn a profit by selling surplus power to the manager. This profit is considered a negative cost, given by
\begin{align}\label{eq: cost in}
J_{\text{s}, i} (t) = - \pi_{\text{in}} \mathbb{E} \left[ \max \left\{ a_i Y_i (t) - X_i (t), 0 \right\} \right]
\end{align}
where $\pi_{\text{in}} ~\mathrm{[\text{\textyen} /kWh]}$ represents the selling price of electricity to the manager. \mml{Note that $\pi_{\text{in}}$ can also be a negative value, analogous to the penalty imposed for reverse power flow into the main power grid in the individual scenario.}

When the power consumption $X_i(t)$ exceeds the power generation $a_i Y_i(t)$ from PV power generation, users can either purchase power from the manager or supplement it with power generated by household fuel cells. 
Here, the following assumptions are made for their selection of power source.
\begin{enumerate}
\item 
The maximum amount of daily purchasable power from the manager for user $i$ is set as $ C_{\text{a}, i}$, where $\sum_{i = 1}^{n} C_{\text{a}, i} = C_{\text{a}}$.
\item
When compensating for power shortages, users will buy electricity from the manager first, and then use household fuel cells if it is still insufficient.
\end{enumerate}
Therefore, for user $i$, the costs of purchasing power from the manager is given by
\begin{align}\label{eq: cost out}
J_{\text{p}, i} (a_i, C_{\text{a},i}, t) = \pi_{\text{out}} \mathbb{E} \left[  \min \left\{ \max \left\{ X_i (t) - a_i Y_i (t), 0 \right\}, C_{\text{a}, i} \right\} \right],
\end{align}
and the cost of getting power from the household fuel cell is given by
\begin{align}\label{eq: cost gas}
J_{\text{gas}, i} (a_i, C_{\text{a},i}, t) = \pi_{\text{gas}} \mathbb{E} \left[ \max \left\{ X_i (t) - a_i Y_i (t) - C_{\text{a}, i}, 0 \right\} \right]
\end{align}
respectively, where $\pi_{\text{out}} ~ \mathrm{[ \text{\textyen}/kWh]}$ represents the selling price per unit of electricity by the manager, and $\pi_{\text{gas}} ~ \mathrm{[\text{\textyen} /kWh]}$ represents the power generation cost per unit of electricity from the household fuel cells. 

Furthermore, the cost of PV panels for user $i$ is expressed by the following,
\begin{align}\label{eq: cost PV}
J_{\text{intp}, i} (a_i) = a_i \pi_{\text{PV}}
\end{align}
where $\pi_{\text{PV}}$ represents the price per unit area of PV panels.

Summarizing, the investment strategy for user $i$ is formulated as an optimization problem
\begin{align}\label{stochastic problem users part}
\min_{a_i} \left\{  J_i (a_i, C_{\text{a},i}) = J_{\text{intp}} + \sum_{t = 1}^{T}  \left\{J_{\text{s}, i} (t) + J_{\text{p}, i} (t) + J_{\text{gas}, i} (t) \right\}  \right\}
\end{align}
\ml{where the arguments of $a_i$, $C_{\text{a},i}$ on the right-hand side are omitted for ease of notation.}
Next, let us consider the cost borne by the manager.  As the manager is responsible for power exchange with all users, the total amount of power bought from users is given by
\begin{align}
P_{\text{in}} (a, t) = \sum_{i = 1}^{n} \max  \left\{ a_i Y_i (t) - X_i (t) , 0\right\}
\end{align}
and power sold to users is given by
\begin{align}
P_{\text{out}} (a, C_{\text{a}}, t) \hspace{-0mm} = \hspace{-0mm} \sum_{i = 1}^{n}  \min \left\{  \max \left\{ X_i (t) \hspace{-0mm} - \hspace{-0mm} a_i Y_i (t), 0\right\},   C_{\text{a}, i} \right\}
\end{align}
respectively.
Thus, the cost of exchanging power with users is given by
\begin{align}
J_{\text{ex}} (a, C_{\text{a}},t) = \mathbb{E} \left[ \pi_{\text{in}} P_{\text{in}} (t) - \pi_{\text{out}} P_{\text{out}} (t) \right].
\end{align}

When the total amount of power $P_{\text{out}} (t)$ sold to users is greater than the total amount of power $P_{\text{in}} (t)$ purchased from users, the manager will extract power from the joint battery storage from the beginning of day $t$ to compensate for the shortage. If the shortage is still not filled, the manager will purchase power from the main power grid with the price $\pi_{\text{grid}} ~ \mathrm{[\text{\textyen} /kWh] }$ per unit of electricity.
Then, the cost in this case is given by
\begin{align}
\hspace{-2mm} J_{\text{grid}} (C_{\text{a}}, t) = \pi_{\text{grid}} \mathbb{E} \left[ \max \left\{ P_{\text{out}} (t) - P_{\text{in}} (t) - \beta_{\text{a}} (t) C_{\text{a}}, 0 \right\} \right]
\end{align}
where $\beta_{\text{a}} (t) C_{\text{a}}$ denotes the power at the beginning of day $t$, with $\beta_{\text{a}} \in [0, 1]$.

On the other hand, when $P_{\text{out}} (t)$ is less or equal to $P_{\text{in}} (t)$, the manager will charge the surplus power to the battery storage. However, if the surplus power exceeds the charging capacity,  the residual power will flow to the main power grid. Then, a penalty for reverse power flow is imposed as in \eqref{eq: cost of rev individual},
\begin{align}
J_{\text{rev}} (C_{\text{a}},t) = \pi_{\text{rev}}\mathbb{E} \left[ \max \left\{ P_{\text{in}} (t) - P_{\text{out} }(t)  - (1 - \beta_{\text{a}} (t) ) C_{\text{a}} ),  0\right\} \right].
\end{align}
Moreover, the investment cost for the battery storage bought by the manager is 
\begin{align}
J_{\text{intb}} (C_{\text{a}}) = \sum_{i = 1}^{n} C_{\text{a}, i} \pi_{\text{B}} .
\end{align}

To sum up, the investment strategy for the manager is formulated into an optimization problem
\begin{align}\label{stochastic problem manager part}
\underset{ C_{\text{a}, i}, \forall i}{\min} \left\{ J_{\text{a}} = J_{\text{intb}} + \sum_{t = 1}^{T} \left\{ J_{\text{ex}} (t) + J_{\text{grid}} (t)  + J_{\text{rev}} (t)  \right\} \right\}.
\end{align}
where the arguments of $a$, $C_{\text{a}}$ are omitted here for ease of notation.
The manager will determine the capacity $C_{\text{a}} = (C_{\text{a}, i}, \ldots, C_{\text{a}, n})$ to invest for given $a_i$, $\forall i$.

The investment strategy is a cooperative game between users and the manager with cost functions defined in \eqref{stochastic problem users part} and \eqref{stochastic problem manager part}, respectively.
The game is called cooperative in the sense that all users will agree to share joint battery storage \cite{chakraborty2018sharing}.
\ml{It will be shown in the next section that the aforementioned cost functions are convex, and Monte Carlo sample average approximation can be applied to solve these problems.}

\section{Stochastic optimal investment strategy}\label{Section Stochastic optimal investment strategy}
The two scenarios introduced in the previous section are formulated into stochastic problems in the following form,
\begin{align}\label{stochastic problem}
\min_{x \in \mathcal{X}} \left\{ f (x) := \mathbb{E} \left[ F(x, \xi) \right] \right\}
\end{align} 
where $\mathcal{X}$ is a nonempty closed subset of $\mathbb{R}^{r}$, $\xi$ is a random variable whose probability distribution is supported on a set $\Xi \subset \mathbb{R}^{d}$, and $F: \mathbb{X} \times \Xi \to \mathbb{R}$ is a function of the two vector variables. We assume that the expectation function $f(x)$ is well-defined and finite-valued for all $x \in \mathcal{X}$.
The formulated problems involve evaluating the expected value of functions of multiple random variables. Thus, it is difficult to solve the stochastic problem analytically. 
In this section, we apply the Monte Carlo sample average approximation (SAA) approach to obtain an approximated solution to problems \eqref{stochastic problem individual users} and \eqref{stochastic problem users part}, \eqref{stochastic problem manager part}.

\subsection{Monte Carlo sample average approximation}\label{subsection: Monte Carlo sample average approximation}

Let us consider a sample $\xi^{1}, \ldots, \xi^{N}$ of $N$ realization of the random vector $\xi$, which can be viewed as data generated by Monte Carlo sampling techniques.  For any $x\in \mathcal{X}$ we can estimate the expected value $f(x)$ by SAA:
\begin{align}\label{sampled optimization problem}
\min_{x \in \mathcal{X}} \left\{ \hat{f}_{N} (x) : = \frac{1}{N} \sum_{j =1}^{N} F(x, \xi^{j})   \right\}.
\end{align}
Assume  problem \eqref{stochastic problem} is feasible and $x^* : = \arg \min_{x \in \mathcal{X}} f(x)$ is the optimal solution to problem \eqref{stochastic problem}. For $\varepsilon \geq 0$, We say that $x \in \mathcal{X}$ is a $\varepsilon$-optimal solution of the problem of minimization of $f(x)$ over $\mathcal{X}$ if 
\begin{align*}
f(x) \leq f(x^*) + \varepsilon.
\end{align*}
Denote $S$, $\hat{S}_N$ as the sets of optimal solutions 
to problem \eqref{stochastic problem} and \eqref{sampled optimization problem}, respectively.  Then,  $S^{\varepsilon}$, $\hat{S}_{N}^{\delta}$ are defined as the sets of $\varepsilon$-optimal solutions, $\delta$-optimal solutions to \eqref{stochastic problem} and \eqref{sampled optimization problem}, respectively.

Denote $D: = \sup_{x, y \in \mathcal{X}} \| x - y\|$ as the diameter of $\mathcal{X}$ and assume it is finite, where $\| x - y\|$ denotes the max-norm of $x-y$.
Assume further that the expectation function $f(x)$ is $L$-Lipschitz continuous on $\mathcal{X}$, that is,  $\left\| f(x) - f(y) \right\| \leq L \| x - y \|$ for all $x, y \in \mathcal{X}$.
Given a significance level $\alpha \in (0, 1)$ for the null hypothesis,
the relation between the sample size $N$ and the probability of the solutions to problem \eqref{sampled optimization problem} being close enough to those of problem \eqref{stochastic problem} is given by the following theorem.

\begin{theorem}[\hspace{1sp}\cite{ruszczynski2003stochastic}]\label{thm: estimating size}
For all $\varepsilon > 0$ small enough, $0 \leq \delta < \varepsilon$,  and $\alpha \in (0,1)$,  and for the sampling size $N$ satisfying 
\begin{align}\label{sample size}
N \geq \frac{12 \sigma^2}{ (\varepsilon - \delta)^2 } \left( r \ln \frac{2DL}{ \varepsilon - \delta } - \ln \alpha \right),
\end{align}
where $r$ is the dimension of $x$, 
\begin{align}\label{variance}
\sigma^2 := \max_{x \in \mathcal{X} / S^{\varepsilon}} \operatorname{\mathbb{V}ar} \left[ F( x^*, \xi) - F( x, \xi)  \right].
\end{align}
\ml{where $\mathcal{X} / S^{\varepsilon}$ is the domain excluding the $\varepsilon$-optimal solution to problem \eqref{stochastic problem}.}
Then, it follows that 
\begin{align}\label{Confidence interval}
\mathbb{P} (\hat{S}_{N}^{\delta} \subset S^{\varepsilon}) \geq 1 - \alpha.
\end{align}
\end{theorem}
Thus, using the Monte Carlo SAA method, the stochastic optimization problem \eqref{stochastic problem} is approximated by the deterministic problem \eqref{sampled optimization problem}.

\subsection{Convexity of cost functions}
We show that cost functions in \eqref{stochastic problem individual users} and \eqref{stochastic problem users part}, \eqref{stochastic problem manager part} are all convex with respect to their decision variables, under reasonable conditions.

\begin{lemma}
The cost function $J_{i} (a_i, C_i)$ in \eqref{stochastic problem individual users} is convex with respect to $a_i$, $C_i$.
\end{lemma}
\begin{proof}
The pointwise maximum and the nonnegative weight sum (as well as the expectation) of functions preserves convexity \cite[p.~79]{boyd2004convex}. Thus, the $J_{\text{s}, i}$ in \eqref{stochastic problem individual users} is convex with respect to both $a_i$ and $C_i$.
\end{proof}
\begin{lemma}\label{lem:convexity of J_a}
The cost function $J_{i}$ in \eqref{stochastic problem users part} is convex with respect to $a_i$, if $\pi_{\text{gas}} \geq \pi_{\text{out}} \geq \pi_{\text{in}}$.
\end{lemma}
\begin{proof}
Since $J_{\text{intp}}$ is linear, we only need to consider the second term on the right-hand side of \eqref{stochastic problem users part}.
Define $g_i (a_i, X_i (t), Y_i (t))$ as 
\begin{align*}
& g_i (a_i, X_i (t), Y_i (t))\\
 =&  \pi_{\text{in}} \max \left\{ a_i Y_i (t) - X_i (t), 0\right\}\\
& +  \pi_{\text{out}}\min \left\{ \max \left\{ X_i (t) - a_i Y_i (t), 0 \right\}, C_{\text{a}, i} \right\}\\
& +\pi_{\text{gas} } \max \left\{ X_i (t) - a_i Y_i (t) - C_{\text{a}, i}, 0 \right\}.
\end{align*}
Then, $J_{\text{s}, i} (t) + J_{\text{p}, i} (t) + J_{\text{gas}, i} (t) = \mathbb{E} \left[ g_i (a_i, X_i (t), Y_i (t))\right]$.
We can observe that $g_i (a_i, X_i (t), Y_i (t) )$ is a piecewise linear function which can be rewritten as 
\begin{align*}
g_i  \hspace{-1mm} = \hspace{-1mm}
\begin{cases}
\pi_{\text{in}} (X_i - a_i Y_i), ~ X_i/Y_i \leq a_i\\
\pi_{\text{out}} (X_i - a_i Y_i),   ~ \frac{X_i - C_{\text{a}, i} }{Y_i} \leq a_i \leq X_i / Y_i\\
 \pi_{\text{out}} C_{\text{a}, i}  + 
\pi_{\text{gas}} (X_i - a_i Y_i - C_{\text{a}, i} ), ~ 0 \leq a_i \leq \frac{X_i - C_{\text{a}, i} }{Y_i}
\end{cases}
\end{align*}
which is shown in \cref{fig: function of gi}.
\begin{figure}[htbp]
\centering
\includegraphics[width = 1\linewidth]{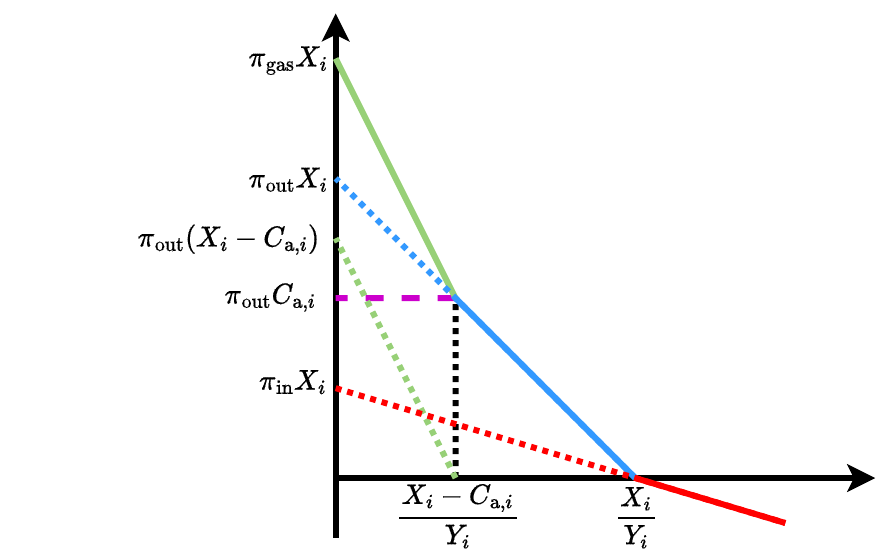}
\caption{Function $g_i(a_u, X_i (t), Y_i (t))$ with respect to $a_i$ when $\pi_{\text{gas}} \geq \pi_{\text{out}} \geq \pi_{\text{in}} > 0$. When $\pi_{\text{in}} \leq 0$, the piece on the right will have a non-negative slope, which does not affect its convexity.}
\label{fig: function of gi}
\end{figure}
We can observe from \cref{fig: function of gi} that when $ \pi_{\text{gas}} \geq \pi_{\text{out}} \geq \pi_{\text{in}}$ holds, the epigraph of $g_i$ is convex and thus $g_i (t)$ is convex with respect to $a_i$.
Next, since the non-negative weight sum of functions preserves convexity \cite[p.~79]{boyd2004convex}, we have that the expectation $J_{\text{s}, i} (t) + J_{\text{p}, i} (t) + J_{\text{gas}, i} (t)$ is convex and the sum of it over $T$ days is also convex with respect to $a_i$. 
\end{proof}

\begin{lemma}
The cost function $J_{\text{a}}$ in \eqref{stochastic problem manager part} is convex with respect to $\left( C_{\text{a},1}, \ldots, C_{\text{a}, n} \right)$ if $\pi_{\text{out}} \geq \pi_{\text{grid}}$.
\end{lemma}
\begin{proof}
Let us focus on the term $J(t) = J_{\text{ex}} (t) + J_{\text{grid}} (t) + J_{\text{rev}} (t)$ in \eqref{stochastic problem manager part}, since $J_{\text{intb}}$ is linear with respect to $\left( C_{\text{a},1}, \ldots, C_{\text{a}, n} \right)$.
$J(t)$ is convex with respect to $(-P_{\text{out}} (t), - C_{\text{a}} (t) )$, as the pointwise maximum and expectation preserve convexity. Meanwhile, if $\pi_{\text{out}} \geq \pi_{\text{grid}}$, we can obtain that $J(t)$ is a monotonically non-decreasing function in each argument, i.e., $-P_{\text{out}} (t)$, and $- C_{\text{a}} (t)$, respectively. Meanwhile, $-P_{\text{out}} (t)$, $- C_{\text{a}} (t)$ are convex with respect to $\left( C_{\text{a},1}, \ldots, C_{\text{a}, n} \right)$. Following from the composition rule \cite[p.~86]{boyd2004convex}, $J(t)$ is convex with respect to $\left( C_{\text{a},1}, \ldots, C_{\text{a}, n} \right)$. Therefore, we have that $J_{\text{a}}$ in \eqref{stochastic problem manager part} is convex. 
\end{proof}

\begin{remark}
The condition $\pi_{\text{gas}} \geq \pi_{\text{out}} \geq \pi_{\text{in}}$ is a reasonable constraint to impose. The first inequality ensures that the user will first purchase electricity from the manager at the price $\pi_{\text{out}}$, and only turn to the fuel cell for power if there is still a shortage, at a higher price $\pi_{\text{gas}}$. The second inequality prevents users from making a profit by engaging in electricity trading, where they buy electricity from the manager at a lower price and sell it back at a higher price. Overall, these constraints promote efficient electricity consumption and discourage wasteful behavior. \ml{The condition $\pi_{\text{out}} \geq \pi_{\text{grid}}$ is also reasonable for the manager as it will not lose money during the process of transmitting power from the grid to users.}
\end{remark}


\begin{remark}
\ml{Since the cost functions in \eqref{stochastic problem users part} and \eqref{stochastic problem manager part} are convex, there exists a Nash equilibrium for the non-cooperative game \cite{facchinei2010generalized}, that is, at the equilibrium, the users and the manager cannot decrease their cost function any further by changing their investments unilaterally.}
\end{remark}

\subsection{Application of the SAA approach}
In this subsection, we discuss the application of the Monte Carlo SAA approach on the formulated stochastic optimization problems to obtain approximated deterministic optimization problems.
To evaluate sample sizes from \Cref{thm: estimating size}, we need to compute the diameters $D$, Lipschitz constants $L$,  and variances $\sigma$ of the corresponding cost functions. The following lemmas provide an estimation of the sample sizes for \eqref{stochastic problem individual users}, \eqref{stochastic problem users part} and \eqref{stochastic problem manager part}, respectively.

\begin{lemma}\label{lem: sample size for scenario one}
Let $a_i \in [0, \bar{a}_i]$, $C_i \in [0, \bar{C}_i]$ for problem \eqref{stochastic problem individual users}, \ml{where $\bar{a}_i, \bar{C}_i$ are the upper bounds for $a_i$ and $C_i$, respectively.} Then, the sampling size $N$ for problem \eqref{stochastic problem individual users} for given $\varepsilon > \delta \geq 0$ and $\alpha > 0$ in \eqref{Confidence interval} should satisfy \eqref{sample size}
with $D = \max \{\bar{a}_i, \bar{C}_i \}$, 
\begin{align*}
L = & \max \left\{ \pi_{\text{PV}} + \max \left\{ \pi_{\text{gas}},  \pi_{\text{rev}}\right\} \sum_{t = 1}^{T} \mathbb{E} \left[Y_i (t)\right],  \right. \\
& \qquad \quad \left. \pi_{\text{B}} + \sum_{t = 1}^{T} \max \left\{ \pi_{\text{gas}} \beta_i (t) , \pi_{\text{rev}}(1 - \beta_i (t)) \right\} \right\}
\end{align*}
and
\begin{align*}
\sigma^2 \leq & D^2
\left( \max \left\{ \max \left\{ \pi_{\text{gas}}, \pi_{\text{rev}} \right\}  \sum_{t = 1}^{T}  \mathbb{E} [Y_i (t)], \right. \right.\\
 & \qquad \left. \left. \sum_{i = 1}^{T} \max \left\{ \pi_{\text{gas}} \beta_i (t), \pi_{\text{rev}} (1 - \beta_i (t)) \right\} \right\} \right)^2.
\end{align*}
\end{lemma}
\begin{proof}
Since $a_i \in [0, \bar{a}_i]$, $C_i \in [0, \bar{C}_i]$, we have that the diameter of the domain is $D = \max \{\bar{a}_i, \bar{C}_i \}$.
Denote $(\tilde{a}_i, \tilde{C}_i)$ as any point in the domain. 
Recall that $J_i (a_i, C_i)$ is a piecewise linear function, 
\ml{and the Lipschitz constant of a piecewise linear function is bounded by the largest Lipschitz constant among all pieces.}
Therefore, we have 
$$\left\| J_i (a_i, C_i) - J_i (\tilde{a}_i, \tilde{C}_i) \right\| \leq L \left\| \begin{bmatrix}
a_i - \tilde{a}_i\\
C_i - \tilde{C}_i
\end{bmatrix} \right\|, $$ where 
\begin{align*}
L = & \max \left\{ \pi_{\text{PV}} + \max \left\{ \pi_{\text{gas}},  \pi_{\text{rev}}\right\} \sum_{t = 1}^{T} \mathbb{E} \left[Y_i (t)\right],  \right. \\
& \qquad \quad \left. \pi_{\text{B}} + \sum_{t = 1}^{T} \max \left\{ \pi_{\text{gas}} \beta_i (t) , \pi_{\text{rev}}(1 - \beta_i (t)) \right\} \right\}.
\end{align*}
Next, recall the property of variance that $\operatorname{\mathbb{V}ar} [F(X) + b] = \operatorname{\mathbb{V}ar} [F(X)]$ for any constant $b$, and $\operatorname{\mathbb{V}ar} [F(X)] \leq \operatorname{\mathbb{E}} [F(X)^2]$. We can obtain from \eqref{variance} that 
\begin{align*}
\sigma^2 \leq & D^2
\left( \max \left\{ \max \left\{ \pi_{\text{gas}}, \pi_{\text{rev}} \right\}  \sum_{t = 1}^{T}  \mathbb{E} [Y_i (t)], \right. \right.\\
 & \qquad \left. \left. \sum_{i = 1}^{T} \max \left\{ \pi_{\text{gas}} \beta_i (t), \pi_{\text{rev}} (1 - \beta_i (t)) \right\} \right\} \right)^2.
\end{align*}
This completes the proof. 
\end{proof}

\begin{lemma}\label{lem: sample size for scenario two a}
Let $a_i \in [0, \bar{a}_i]$ for problem \eqref{stochastic problem users part}, then the sampling size for problem \eqref{stochastic problem individual users} for given $\varepsilon > \delta \geq 0$ and $\alpha > 0$ in \eqref{Confidence interval} should satisfy \eqref{sample size} with $D = \bar{a}_i$,
$
L = \pi_{\text{PV}} + \pi_{\text{gas}}  \sum_{t = 1}^{T}  \mathbb{E} [Y_i (t)],
$
and 
\begin{align*}
\sigma^2 = D^2\pi_{\text{gas}}^2 \left( \sum_{t = 1}^{T} \mathbb{E} [Y_i (t)] \right)^2.
\end{align*}
\end{lemma}
\begin{proof}
\ml{The largest Lipschitz constant of the cost function is bounded by $L = \pi_{\text{PV}} + \sum_{t = 1}^{T} \pi_{\text{gas}} \mathbb{E} \left[ Y_i (t)\right]$ following from \Cref{fig: function of gi} in the proof of \Cref{lem:convexity of J_a}.
Following an analogous estimation, we can obtain the variance. }
\end{proof}

\begin{lemma}\label{lem: sample size for scenario two Ca}
Let $C_{\text{a}, i} \in [0, \bar{C}_{\text{a}, i}]$, $\bar{C}_{\text{a}, i} \geq 0$, $i = 1, \ldots, n$, for problem \eqref{stochastic problem manager part}, then the sampling size for problem \eqref{stochastic problem manager part} for given $\varepsilon > \delta \geq 0$ and $\alpha > 0$ in \eqref{Confidence interval} should satisfy \eqref{sample size} with $D = \|\bar{C}_{\text{a}}\|$,
\begin{align*}
L =  \pi_{\text{B}} + \sum_{t = 1}^{T} \max \left\{ \pi_{\text{out}}, \pi_{\text{grid}} (1+\beta_{\text{a}} (t)), \pi_{\text{rev}} \beta_{\text{a}} (t) \right\}
\end{align*}
and
$$
\sigma^2 = D^2 \left( \sum_{t = 1}^{T} \max \left\{ \pi_{\text{out}}, \pi_{\text{grid}} (1+\beta_{\text{a}} (t)), \pi_{\text{rev}} \beta_{\text{a}} (t) \right\} \right)^2.
$$
\end{lemma}
The proof is analogous to the above lemmas and is thus omitted here.

\section{A case study}\label{Section A case study}
We implement the optimal investment strategy for PV and battery storage in the Jono neighborhood mentioned in \Cref{subsection jono}.
In this neighborhood, there are $140$ households and an operation period of $T = 334$ days is considered. The daily power generation and power consumption for the users are represented as random variables, generated based on their historical data (see, e.g., \Cref{fig:Power consumption and generation}). The price for PV is set to $\pi_{\text{PV}} = 2000 ~ \mathrm{[ \text{\textyen} /m^2 ]}$
and the price for battery storage is set to $\pi_{\text{B}} = 4500 ~ \mathrm{[ \text{\textyen}/kWh  ]}$. These prices are amortized over a ten-year lifespan.
We set the electricity price at $\pi_{\text{gas}} = 30 ~ \mathrm{[\text{\textyen} /kWh] }$, and the penalty for reverse power flow at $\pi_{\text{rev}} = 20 ~ \mathrm{[\text{\textyen} /kWh] }$. The monetary unit \text{\textyen} indicates JPY.
The charge rate of the battery is set to $\beta = 0.5$, ensuring it has the capacity to charge or discharge every day.

In the first scenario, users invest in PV panels and battery storage individually and we consider the problem represented by \eqref{stochastic problem individual users}. 
The stochastic problem is solved by applying the Monte Carlo approach introduced in \Cref{subsection: Monte Carlo sample average approximation}.
Figure \ref{fig:comparison between investment plans} depicts a comparison before and after the optimal planning and the optimal individual investment plans. It should be noted that no battery storage was utilized in the area before the implementation of the optimal planning. The results show that introducing battery storage allows users in the Jono neighborhood to invest less in PV panels while significantly reducing the annual cost.
\begin{figure}[htbp]
    \centering
\begin{subfigure}[b]{1\linewidth}
     \centering
     \includegraphics[width = 1\linewidth]{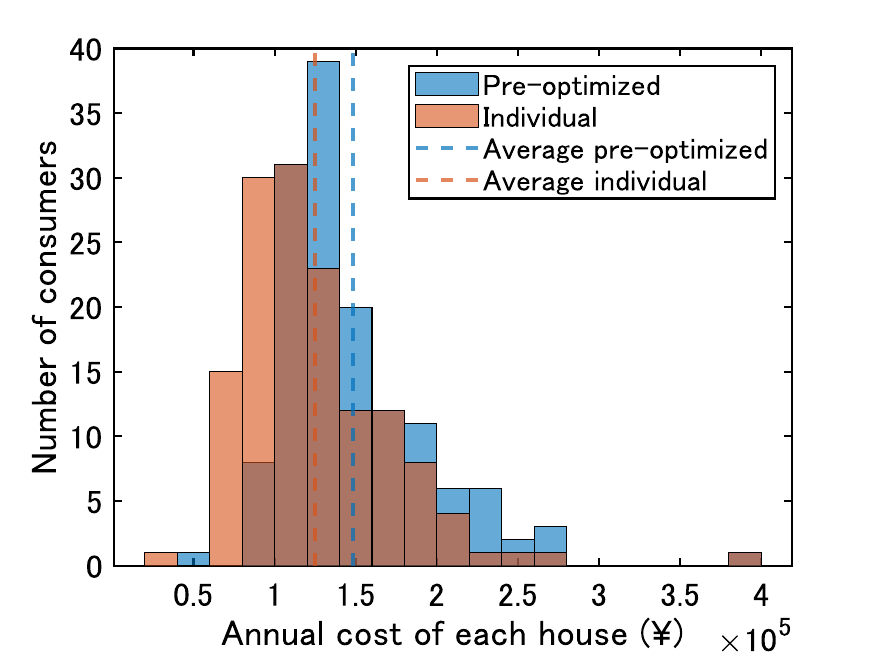}
     \caption{Comparison between the total annual cost before optimization and that after optimal individual investment for all users.}
     \label{fig:Investment plan cost comparison}
\end{subfigure}
\begin{subfigure}[b]{1\linewidth}
     \centering
     \includegraphics[width = 1\linewidth]{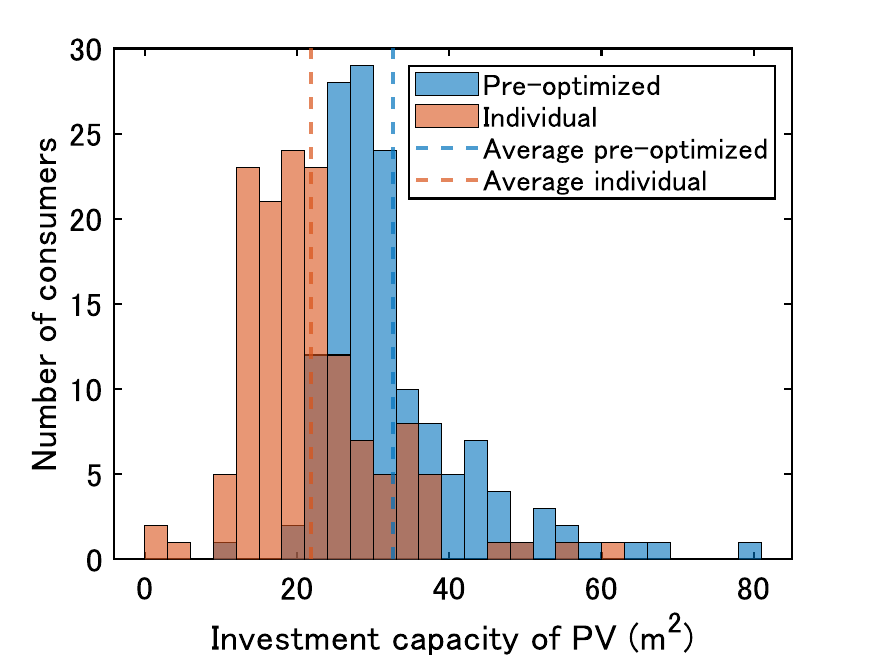}
     \caption{Comparison between the PV investment before optimization and that after optimal individual investment for all users.}
     \label{fig:Investment plan pv panels}
\end{subfigure}
\begin{subfigure}[b]{1\linewidth}
     \centering
     \includegraphics[width = 1\linewidth]{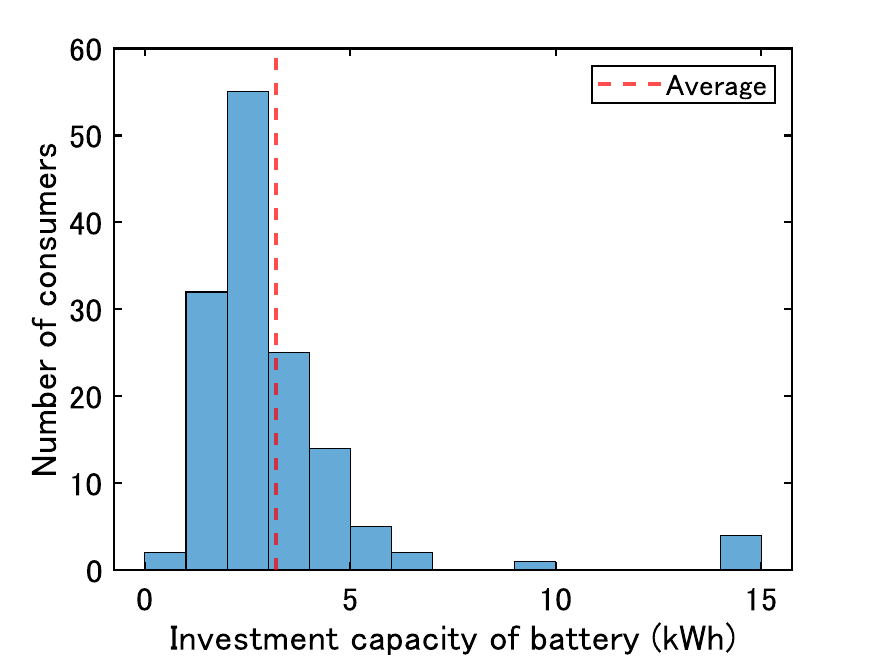}
     \caption{Optimal individual investment plan for battery storage for all residents.}
     \label{fig:Investment plan for battery individual users}
\end{subfigure}
    \caption{The investment plans for all users before and after optimal investment planning}
    \label{fig:comparison between investment plans}
\end{figure}

Next, we consider the joint investment strategy in storage.
First, let us examine the global optimization (\ref{eq:global_optimization_costfun})--(\ref{eq:global_optimization_Jrev}), assuming the same household data and equipment/electricity prices as in the individual case. It is important to note that this global optimization does not specify how the joint battery storage is distributed among individuals. Therefore, in the second scenario, we adopt a game setting.
The manager determines the allocation of battery capacity, calculated through the manager's optimization problem \eqref{stochastic problem manager part}, for all users. Each user then utilizes their respective allocation to decide on the invested area of PV panels.
We set the electricity price at $\pi_{\text{grid}} = 10 ~\mathrm{[\text{\textyen} /kWh] }$, and the penalty for reverse power flow at  $\pi_{\text{rev}} = 20 ~ \mathrm{[\text{\textyen} /kWh] }$.
The selling price by the manager is $\pi_{\text{out}} = 20 ~ \mathrm{[ \text{\textyen}/kWh]}$, and the gas price is $\pi_{\text{gas}} =  30 ~ \mathrm{[\text{\textyen} /kWh]}$. We consider two cases in this game model, namely $\pi_{\text{in}} = 5 ~ \mathrm{[\text{\textyen} /kWh]}$, representing the buy-in price from users, and $\pi_{\text{in}} = -5 ~ \mathrm{[\text{\textyen} /kWh]}$ indicating a penalty to users for excessive power generation.
The investment strategy is the Nash equilibrium, which can be obtained by iteratively solving the optimization on the user's side \eqref{stochastic problem users part} and the manager's side \eqref{stochastic problem manager part}. 

The comparison of simulation results among different models in the community is shown in \Cref{fig:comparison among investment plans}. 
The individual model corresponds to the first scenario \eqref{stochastic problem individual users}, while the game models correspond to the second scenario \eqref{stochastic problem users part}, \eqref{stochastic problem manager part}. On the other hand, the global model represents the ideal case \eqref{eq:global_optimization_costfun}, where the entire neighborhood is treated as a single user.
\Cref{fig:joint model cost comparison} shows the comparison of the total annual costs in the community with the past data. Both the individual and joint scenarios result in lower total costs compared to the non-optimized scenario. It is expected that the total costs in both scenarios are higher than the lower bound obtained from the global model. This gap arises due to the lack of sharing for the first scenario and the conflicting interests among individual users for the second scenario.

\mml{It is worth noting that the total cost from the individual model is lower than that of the joint investment game model when $\pi_\text{in} = 5~\mathrm{[\text{\textyen} /kWh]}$.
This situation arises as users can profit by selling their excess electricity to the manager, thereby leading to an amplified investment in PV panels to maximize this financial advantage. However, this over-investment ultimately escalates the manager's expenses, resulting in elevated overall costs.
To solve the above issue, let the excessive power generation be penalized similarly to the individual model. Note that the penalty is less than $\pi_{\text{rev}} = 20 ~ \mathrm{[\text{\textyen} /kWh]}$, which still provides users an incentive to participate in the joint investment. The simulation result in \Cref{fig:joint model cost comparison} shows a decrease in total cost.
Intriguingly, this discovery compels us to reassess the intuitive understanding of the sharing economy. It underscores the necessity for an appropriate sharing strategy to prevent users from over-exploiting communal resources, which, without careful management, could lead to inflated overall costs in the long run. It thus highlights the challenges of implementing the joint model, as it requires a significant effort to reconcile the competing interests of users.}
Furthermore, \Cref{fig:total pv comparison} illustrates the total investment amount for PV panels in the community, indicating a clear incentive to invest in PV panels when energy recycling from users is taken into account. \Cref{fig:total battery comparison} depicts  the total investment amount for battery storage in the community, which suggests that sharing energy can help restrain redundant battery investments.
\begin{figure}[htbp]
    \centering
\begin{subfigure}[b]{1\linewidth}
         \centering
         \includegraphics[width = 1\linewidth]{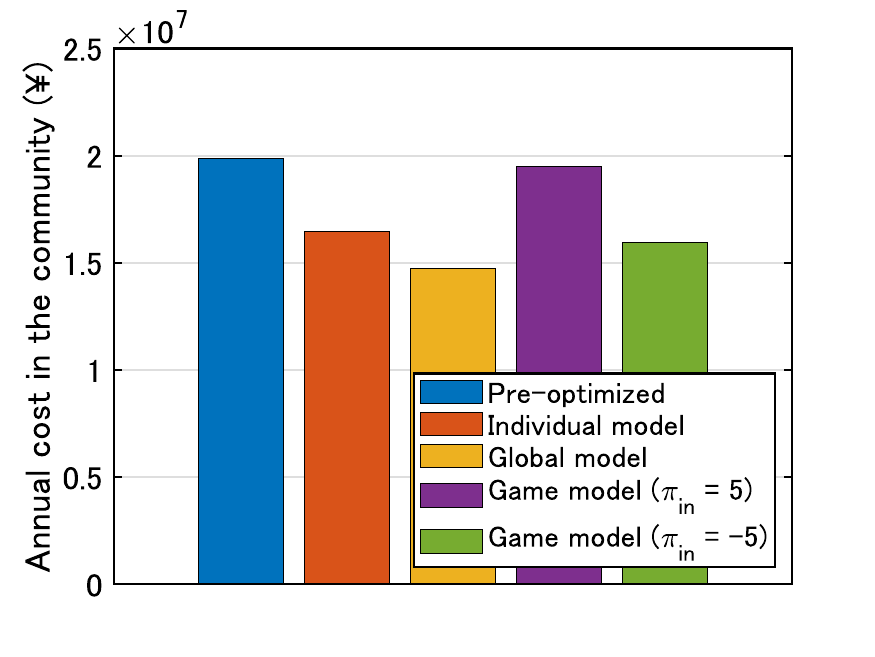}
         \caption{Comparison of the total annual cost in the community among four models with the past data (pre-optimized, individual model, global optimization model, game scenario).}
         \label{fig:joint model cost comparison}
    \end{subfigure}
\begin{subfigure}[b]{1\linewidth}
         \centering
         \includegraphics[width = 1\linewidth]{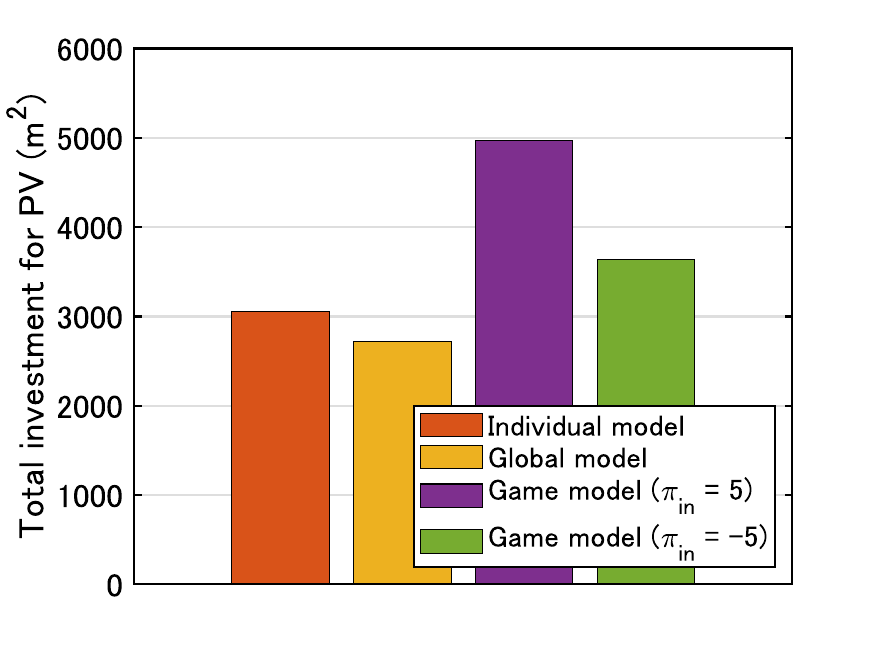}
         \caption{Comparison of the total investment for PV panels among three models.}
         \label{fig:total pv comparison}
    \end{subfigure}
\begin{subfigure}[b]{1\linewidth}
     \centering
     \includegraphics[width = 1\linewidth]{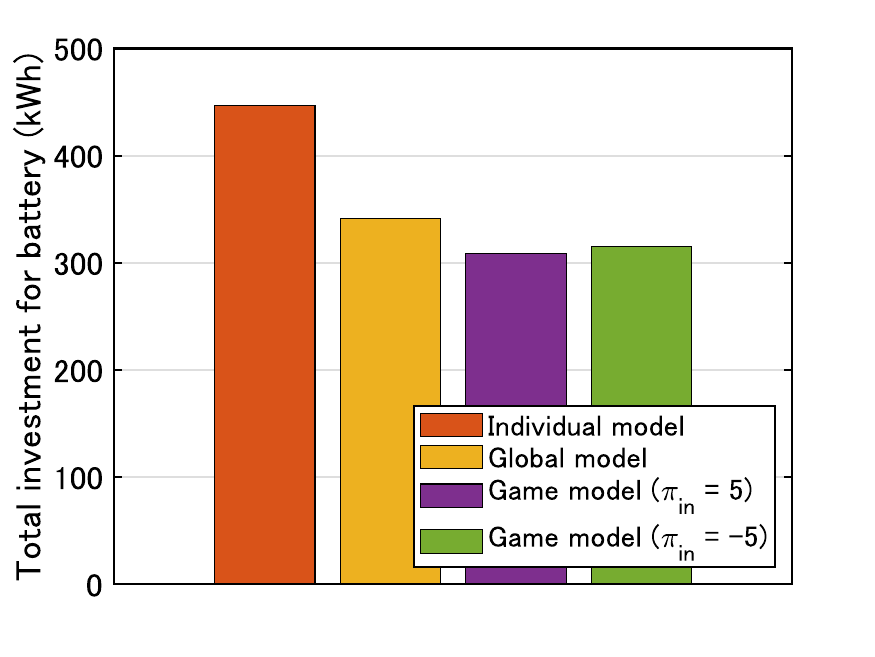}
     \caption{Comparison of total investment for battery storage among three models.}
     \label{fig:total battery comparison}
\end{subfigure}
    \caption{The the total annual costs, PV panels and battery capacity investment in the community for all models.}
    \label{fig:comparison among investment plans}
\end{figure}

\ml{All the stochastic problems are approximated as deterministic problems through the Monte Carlo approach. Let $\varepsilon = 5000 \mathrm{[\text{\textyen}]}$ and the significance level $\alpha = 0.01$. The estimation of the sample sizes yields $N = 2 \times 10^{6}$ for \Cref{lem: sample size for scenario one}, $N = 3 \times 10^{3}$ for \Cref{lem: sample size for scenario two a} and $N = 1.3 \times 10^{5}$ for \Cref{lem: sample size for scenario two Ca}, respectively.
We remark that conservative estimations of the variances are carried out in Lemma~\ref{lem: sample size for scenario one}--\ref{lem: sample size for scenario two Ca} and thus only conservative estimates of the required sample sizes are provided. In practice, much smaller sample sizes are sufficient to guarantee accurate results. However, finding a tight sample size is not our goal in this paper.
}

\section{Conclusion}\label{Section Conclusion}
In this paper, we have explored the investment strategy for PV panels and battery storage to attain the Net-Zero Energy House status within a regional power system comprising a manager and multiple users.

We have demonstrated through a case study in Kitakyushu that incorporating battery storage into the power system effectively reduces power imbalances and enhances energy utilization efficiency, which is crucial for attaining ZEH objectives. 
Furthermore, our analysis of the two proposed scenarios has revealed their potential to significantly decrease annual electricity costs.

Additionally, we have highlighted the importance of implementing a proper sharing policy to incentivize individual users to share electricity while preventing the excessive exploitation of communal resources. Our future work will concentrate on formulating games that meet these requirements and foster the desired outcomes.

\section*{Acknowledgment}
This work is supported by the Ministry of the Environment, Government of Japan.

\bibliographystyle{iet}
\bibliography{References}

\end{document}